\documentclass[%
 reprint,
 amsmath,amssymb,
 aps,
]{revtex4-1}

\usepackage{graphicx}
\usepackage{dcolumn}
\usepackage{bm}
\usepackage{amssymb}

\newtheorem{theorem}{Theorem}
\newtheorem{proof}{Proof}

\newtheorem{definition}{Definition}
\newtheorem{lemma}{Lemma}
\newtheorem{lproof}{Proof}

\begin{document}


\title{The Null and Force-Free Electromagnetic Field}

\author{Govind Menon}
\affiliation{Department of Chemistry and Physics\\ Troy University, Troy, Al 36082}

\date{\today}

\begin{abstract}
This paper describes the electrodynamics of a null and force-free field in completely geometric terms. As was previously established in \cite{Menon_FF20}, solutions to force-free electrodynamics are governed by the existence of certain special types of foliations of spacetime. Here the nature of the foliations in a coordinate-free formalism in the null case is prescribed. All of the general results are illustrated by constructing a null, force-free electrodynamic field in a Friedmann-Lemaitre-Robertson- Walker (FLRW) spacetime.
\end{abstract}

\pacs{Valid PACS appear here}
\maketitle


   \section{Introduction}
   Force-free electrodynamics (FFE) has found applications in a wide range of astrophysical phenomena ever since Blandford and Znajek published their seminal article describing the powering mechanism of black holes (\cite{BZ77}). FFE is relevant when the energy density of the interacting plasma is insignificant compared to the electromagnetic field density so that the transfer of energy between the two can be effectively ignored. Most of the effort in constructing solutions to the relevant equations are numerical in nature (for example see \cite{K04}, \cite{RMTASS}, \cite{Koide_2019} and \cite{Qian_2018}). Dynamic evolution of the electromagnetic field in varying astrophysical settings can be non-trivial, and numerical solutions will continue to play a primary role in applications. However, a theoretical framework can provide physical meaning to the equations and guide the search for new solutions.
   
   As shown in \cite{Menon_FF20}, FFE is completely determined by gravity (geometry) alone. In the non-null case, (i.e., when the field is electrically or magnetically dominated), picking initial conditions is restricted to a choice of an integration constant in a suitable coordinate system. In the null case, there will always be a class of solutions that has exactly two degrees of freedom. The article will focus only on the null solutions of the electromagnetic field. The degrees of freedom in the null case will be clarified. Further, it will shown that the associated foliations/field sheets will always contain a null geodesic. And finally, a precise relationship between the null expansion scalar of the congruence and the field sheets will also be established.
   
   Early research on force-free electromagnetic theory showed that the field tensor can be written as a simple 2-form, and further that the kernel of the field forms an involutive distribution (\cite{Carter79}, \cite{Uchida1}, \cite{Uchida2}). Over a decade ago, the first exact analytical solution to the Blandford-Znajek equations was obtained using a 3+1 decomposition of the electromagnetic field (\cite{MD07}). This solution described a null, stationary, axis-symmetric field in a Kerr background. The same technique led to a second solution, however this time, the solution was magnetically dominated (\cite{Menon15}). Using the Newmann-Penrose formalism, Brennan et. al. were able to generalize the original null solution to a non-axis-symmetric and time-dependent case (\cite{BGJ13}). In this paper, it will be clarified why such a generalization necessarily exists, and why no further generalization of a null field generated by the outgoing principal null geodesic of the Kerr geometry exists. 
   
   This work begins with a geometric recasting of the recent paper, \cite{Menon_FF20}, that describes the structural aspects of FFE in albeit an adapted coordinate chart. The adapted chart has the advantage of simplifying the relevant partial differential equations in a way that one can study initial data surfaces and the evolution of the field. It is also clear that null and non-null fields have very different characteristics. The adapted chart formalism in \cite{Menon_FF20} partially masked the inherent geometric nature of the theory. The geometric formulation in the null case will connect FFE to null geodesic congruences and its associated null mean curvature. All the central results derived in this paper will be illustrated by constructing a new, non-trivial, null, force-free solution in a $k=+1$ FLRW background.
\section{The Force-Free Electromagnetic field}
As per general relativity, spacetime is a 4-dimensional smooth manifold ${\cal M}$ endowed with a metric $g$ of Lorentz signature which we choose as $(-1, 1,1,1)$. In this work, the metric is fixed and predetermined. The only restriction is that the background metric is free of any electromagnetic contribution. The electromagnetic field tensor can be written as a $2$-form $F$ which satisfies the following Maxwell's equations:
\begin{equation}
 d  F = 0\;,
 \label{Fclosed}
\end{equation}
 and
\begin{equation}
* \;d * F = j\;.
\label{inhomMaxform}
\end{equation}
Here $*$ is the Hodge-Star operator and $d$ is the exterior derivatives on forms. Also, $j$ denotes the current density vector. Force-free electrodynamics is a restricted case where we place the following additional requirement
$$i_{j^\sharp}\; F =0\;.$$
Here $i$ denotes the interior product defined by
$$i_{j^\sharp}\; F \equiv F(j^\sharp, \cdot )\;,$$
where $j^\sharp$ is the contravariant vector field given by the map 
$$j^\sharp = g^{\mu\nu} \;j_\nu \;\partial_\mu$$
in any coordinate basis. The $\sharp$ operator can be used to convert any 1-form to a tangent vector. Its inverse will be denoted by the $\flat$ operator which maps tangent vectors to 1-forms; for example given a vector field $\chi = \chi^\mu \;\partial_\mu$, define
$$\chi^\flat \equiv g_{\mu\nu} \chi^\nu dx^\mu\;.$$
The above expression is also valid in any local chart. 
When using the abstract index notation the $\flat$ and $\sharp$ superscripts will be suppressed, i.e., $j^\mu$ will be understood to be $j^\sharp$, and $\chi_\mu$ is to be taken as $\chi^\flat$. For any $p\in {\cal M}$, let
$$F^2 (p)\equiv F_{\mu\nu} F^{\mu\nu} (p)\;.$$
Then
$F$ is said to be magnetically dominated at $p$ whenever $F^2 (p)>0$, $F$ is  electrically dominated at $p$ whenever $F^2 (p)<0$, finally a  force-free electromagnetic field $F$ is  null at $p$ whenever $F^2 (p)=0$. The 3 current $J \equiv d*F$. Then  $*J=j$.

\vskip0.2in
Earlier works by  \cite{Carter79}, \cite{Uchida1}, and \cite{Uchida2} shows that a force-free electromagnetic field can always be written as a simple 2-form:
\begin{equation}
    F= \alpha \wedge \beta\;.
    \label{simF}
\end{equation}
Here $\alpha$ and $\beta$ are 1-forms. The recent paper by \cite{GT14} explains all the essential equations of FFE listed below. The kernel of $F$, denoted by $\ker F$, is a 2-dimensional subspace of the tangent bundle satisfying the property that $i_v F =0$ whenever $v \in \ker F $.
For the force-free case, using eq.(\ref{simF}) it is easily shown that $\ker F$ at $p$ is spacelike/Lorentz whenever $F$ is electrically/magnetically dominated. When a force-free $F$ is null the metric when restricted  to $\ker F$ is degenerate. For $v, w \in \ker F$,
$$i_{[v,w]} F= [{\cal L}_v, i_w] F =0$$
since $F$ is a closed 2-form. Here ${\cal L}_v$ denotes the Lie derivative with respect to the vector field $v$.
Therefore $\ker F$ is an involutive distribution, meaning that whenever vector fields $v,w \in \ker F$, then $[v,w] \in \ker F$. Consequently, Frobenius' theorem implies that when a force-free $F$ exists on ${\cal M}$, spacetime can be foliated by 2-dimensional integral submanifolds of the distribution spanned by $\ker F$  (\cite{JLee13} covers distributions, Foliations, and Frobenuis' theorem in a very readable manner). The leaves of the foliation, which are the integral submanifolds of $\ker F$,  will be denoted as $ {\cal F}_a$. Here $a$ belongs to some indexing set $A$. The key points here are that
$$ {\cal F}_a \cap {\cal F}_b =0, \;{\rm whenever}\;a \neq b \in A\;,$$
$$ \cup_{a\in A} \;{\cal F}_a ={\cal M}\;,$$
and whenever $v\in T({\cal F}_a) $ for any $a \in A$ we have that $i_v F=0$. When $F$ is written in the form given by eq.(\ref{simF}), the force-free condition is equivalent to the prescription
\begin{equation}
    J \wedge \alpha = 0 =J \wedge \beta\;.
    \label{FFwedgecond}
\end{equation}
Following Gralla and Jacobson, ${\cal F}_a$ will be referred to as field sheets.
\section{The Force-Free Field in an Adapted Chart}
In this section, the relevant results in \cite{Menon_FF20} are presented to maintain continuity of discussion. About any $p$ in some ${\cal F}_a$, there exist an adapted coordinate chart $\big(U_p, \phi_p= (x^1, \dots, x^4)\big)$ centered about $p$, i.e.,
$\phi_p (p) = 0\;,$ such that the slices given by constant values of 
$x^3$ and $x^4$ are indeed the field sheets contained in $U_p$. Consequently, as shown in \cite{GT14}, the Maxwell field tensor can be written as
\begin{equation}
F= u(x^3, x^4) \;dx^3 \wedge dx^4\;.
\label{finalformF}
\end{equation}
Any such chart $\big(U_p, \phi_p= (x^1, \dots, x^4)\big)$ with the above mentioned properties is referred to as a {\bf field sheet adapted chart} for $F$. As previously mentioned in \cite{Menon_FF20}, there is no preference here for a timelike coordinate, and so adapted coordinates are labelled with indices ranging from $1-4$, rather than the usual $0-3$. In the adapted chart, define quantities 
$$M^r= g^{r 3} \;g^{3 4}- g^{3 3} \;g^{r 4}\;,\; {\rm and}\;\;N^r= g^{r 3} \;g^{4 4}- g^{3 4} \;g^{r 4}\;,$$
for $r=1-4$.
It was shown in \cite{Menon_FF20} that the equations of FFE are then given by
\begin{equation}
  M^4 \;\frac{\partial}{\partial{x^4}} \ln |u| = - \frac{1}{\sqrt{-g}}\; \frac{\partial}{\partial{x^r}} \left(\sqrt{- g}\; M^r\right)\equiv -\nabla_r M^r\;, 
  \label{FFEE1}
\end{equation}
and
\begin{equation}
   N^3 \;\frac{\partial}{\partial{x^3}} \ln |u| =  -\frac{1}{\sqrt{-g}}\; \frac{\partial}{\partial{x^r}} \left(\sqrt{- g}\; N^r\right)\equiv- \nabla_r N^r\;. 
   \label{FFEE2}
\end{equation}
Note that $ M^4= -N^3$, and $M^3=0=N^4$. In the last two equations, we have pretended that $M^r$ and $N^r$ are vector fields in writing the shorthand notation for divergence. For the null case, this assumption will be justified along the way.   
\vskip0.2in
Let $F$ be a null electromagnetic field on ${\cal M}$, and let ${\cal F}_a$ be the associated field sheets. Then, as mentioned in the previous section, $g$ restricted to the  tangent bundle of ${\cal F}_a$ denoted by $T({\cal F}_a)$ is degenerate. To see how this happens, consider the plane spanned by $\alpha$ and $\beta$. By a judicious reassignment if necessary, pick $\alpha \perp \beta$ in the sense that $g(\alpha, \beta)=0.$ Since,
$$F^2=2 \alpha^2 \beta^2 =0\;,$$
either $\alpha$ or $\beta$ must be a null vector. Without loss of generality, set $\beta^2=0$, and thus $\alpha$ is spacelike, and $\ker F$ consists of all vectors orthogonal to $\alpha$ and $\beta$. In particular $\beta^\sharp \in \ker F$. Rename $\beta^\sharp$ as $l$ to indicate that it is a lightlike vector. Moreover
$$g(l,v)= \beta (v)=0, \forall v \in T({\cal F}_a)\;.$$
Hence, $g$ when restricted to $T({\cal F}_a)$ is degenerate. The exact same argument shows that $g$ when restricted to the plane spanned by $\alpha$ and $\beta (= l^\flat)$ is degenerate. In the adapted chart, this means that 
\begin{equation}
   \det \left(
    \begin{array}{cc}
      g^{33} & g^{34} \\
      g^{43} & g^{44} \\
    \end{array}
  \right) =0\;, 
  \label{detkerform} 
\end{equation}
or equivalently $M^4 = 0=N^3$. Henceforth,  foliations by 2-dimensional submanifolds of ${\cal M}$, where the restriction of $g$ on the leaves of the foliation is degenerate, will be referred to as a  {\bf null foliation}. 
Given a null foliation, there may or may not be an associated force-free null field. However, foliation adapted charts, $\big(U_p, \phi_p= (x^1, \dots, x^4)\big)$, such that surfaces of constant values for $x^3$ and $x^4$ that agree with the leaves of the foliation continue to exist. Such charts are referred to as a {\bf null foliation adapted chart}. The following theorem was the first insight that led to the understanding of null force-free solutions in an arbitrary spacetime (\cite{Menon_FF20}).

\begin{theorem} {\rm({\bf Version} 1)}
Let ${\cal F}$ be a null foliation of ${\cal M}$ with metric $g$.  Let $\big(U_p, \phi_p= (x^1, \dots, x^4)\big)$ be a null foliation adapted chart about any arbitrary point $p \in {\cal M}$. Then $F$ given by eq.(\ref{finalformF}) for any smooth function $u(x^3,x^4)$ is a unique class of force-free, null solution in $U_p$  such that $\ker F$ contains exactly the vectors in $T({\cal F}_a)$ if and only if
\begin{equation}
  \nabla_r M^r = 0 = \nabla_r N^r\;.  
  \label{divfreeFFeq}
\end{equation}

\label{existuninull}
\end{theorem}
\begin{proof} Since $M^4=0=N^3$, the result follows immediately from  eqs.(\ref{FFEE1}) and (\ref{FFEE2}).
$\blacksquare$
\end{proof}
In the following section, the above theorem will be rewritten in completely geometric terms. In doing so, the physical meaning behind the requirements of the theorem, and its chart independence will be manifest. In the meantime consider another null foliation adapted chart $(\bar x^1, \dots, \bar x^4)$. It is instructive to see how a simple 2-form of the type
$$u(x^3, x^4) \;dx^3 \wedge dx^4$$
might transform under a change of adapted coordinates. Since the new chart is also adapted to the foliation, we must have that
\begin{equation}
   \frac{\partial x^a}{\partial \bar x^i}=0\;.
   \label{charttrans}
\end{equation}
Here $i=1,2$ and $a=3,4$.
Then
$$\frac{\partial}{\partial \bar x^i} u(x^3, x^4) = \frac{\partial x^1 }{\partial \bar x^i}\; \frac{\partial }{\partial  x^1} u(x^3, x^4) + \frac{\partial x^2 }{\partial \bar x^i} \;\frac{\partial }{\partial  x^2} u(x^3, x^4) =0\;.$$
Also
$$dx^3 \wedge dx^4= D\;d\bar x^3 \wedge d\bar x^4$$
where
\begin{equation}
   D =\left(\frac{\partial x^3}{\partial \bar x^3}\;\frac{\partial x^4}{\partial \bar x^4}-\frac{\partial x^4}{\partial \bar x^3}\frac{\partial x^3}{\partial \bar x^4}\right)\;. 
   \label{ddef}
\end{equation}
From eq.(\ref{charttrans})
$$\frac{\partial }{\partial \bar x^i} \frac{\partial x^a}{\partial \bar x^b} =\frac{\partial }{\partial \bar x^b} \frac{\partial x^a}{\partial \bar x^i}=0\;,$$
for $a,b = 3,4$ and $i=1,2$, and so
$$\frac{\partial }{\partial \bar x^i} D=0\;.$$
I.e., we get that the original 2-form
$$u(x^3, x^4) \;dx^3 \wedge dx^4 =\bar u(\bar x^3, \bar x^4) \;d\bar x^3 \wedge d\bar x^4 $$
where 
\begin{equation}
    \bar u(\bar x^3, \bar x^4) = D \cdot u(x^3, x^4)\;.
    \label{utrans}
\end{equation}
I.e., $F$ preserves form under a change of null foliation adapted charts.
\section{The Geometry of the Null Force-Free Field}

The ray along $l$ generate all the null vectors in $T({\cal F}_a)$. In this manner, a null force-free field $F$ defines a unique null ray in spacetime that are tangent to the field sheets; i.e., locally, one obtains a smooth lightlike vector field $l$ in $\ker F$. Since a null vector cannot be normalized, $l$ is far from unique. But this is a familiar problem in the theory of null hypersurfaces and we know how to deal with this redundancy. Globally, $l$ defines a null congruence. At the onset, there is no reason to assume that $l$ is a geodesic congruence. Since, ${\cal F}_a$ is 2-dimensional, there exists a local spacelike vector field $s$ in $\ker F$. We will normalize $s$ so that $g(s,s)=1$. Together $l$ and $s$ span $\ker F$. Note that
$\ker F$ can also be defined by the requirement 
\begin{equation}
   \ker F= \Big \{v \in T({\cal M})\; | \;\alpha(v)=0=l^\flat(v)\;\Big\}\;.
   \label{kerf}
\end{equation}
Finally, normalize $\alpha$ so that $g(\alpha, \alpha)=1$, and to complete the tetrad, let $n$ be a null vector field such that
$$n^\flat (l)=-1, \;\;n^\flat (s)=0=\alpha(n)\;.$$
To recap, $l$ and $s$ span $\ker F$, the span of $l^\flat$ and $\alpha$ are all the forms that annihilate vectors in $\ker F$, $n$ completes the tetrad, and $(s, l, \alpha^\sharp, n)$ span $T({\cal M})$. We shall refer to $(s, l, \alpha^\sharp, n)$ as a {\bf null foliation adapted frame} for a null foliation ${\cal F}$. $l,n$ are null, and $\alpha^\sharp, s$ are unit spacelike vectors such that
$$l^\perp = {\rm span}\;\{\;l, \;s, \;\alpha^\sharp\;\}\;,$$
and
$$n^\perp = {\rm span}\;\{\;n, \;s, \;\alpha^\sharp\;\}\;.$$
There is not a unique way to pick out $l,s$ and $\alpha$. One is free to make transformation of the following type,
\begin{equation}
    l \rightarrow f\; l,\; {\rm wherein}\; n \rightarrow (1/f)\; n\;,
    \label{lntrans}
\end{equation}
and
\begin{equation}
    s \rightarrow s+ g\;l\;,
    {\rm and} \;
    \alpha \rightarrow \alpha + h \;l^\flat\;,
    \label{satrans1}
\end{equation}
wherein
$$n \rightarrow n +\left(\frac{g^2+h^2}{2}\right)\;l + g\;s+h\;\alpha^\sharp\;,$$
and all the essential required properties of $l,s$ and $\alpha$ are still retained.
Here $f \neq 0,\;g$ and $h$ are any smooth function on ${\cal M}$. 
The differential-forms version of Frobenius's theorem, eq.(\ref{kerf}), and the fact that $\ker F$ is integrable implies that 
\begin{equation}
    dl^\flat= l^\flat \wedge A + \alpha \wedge B\;,
    \label{dflat}
\end{equation}
and
\begin{equation}
    d\alpha= l^\flat \wedge C + \alpha \wedge D\;,
    \label{flatal}
\end{equation}
for some 1-forms $A,B,C$ and $D$. We can agree to write $A$ as
$$A= A_n \;n^\flat+ A_\alpha\; \alpha + A_s \;s^\flat\;,$$
and $B$ as
$$B= B_n \;n^\flat+ B_s \;s^\flat\;.$$
Similar remarks apply to $C$ and $D$. There is a choice between eq.(\ref{simF}) and eq.(\ref{finalformF}) as a starting point of our discussion. Although we plan to proceed in a completely geometric, coordinate free formalism, it will be useful to establish a relationship between the two expressions. Let

\begin{equation}
  \left(
           \begin{array}{c}
             \alpha \\
             l^\flat \\
           \end{array}
         \right)=\left(
    \begin{array}{cc}
      \alpha_3 & \alpha_4 \\
      l^\flat_3 & l^\flat_4  \\
    \end{array}
  \right)\left(
           \begin{array}{c}
             dx^3 \\
             dx^4 \\
           \end{array}
         \right)\;. 
         \label{chart2frame}
\end{equation}
Then
$$u(x^3, x^4) \;dx^3 \wedge dx^4 = (u \cdot \kappa)\; \alpha \wedge l^\flat\;,$$
where 
\begin{equation}
  \kappa=(\alpha_3 \;l^\flat_4-\alpha_4\; l^\flat_3)^{-1}\;.
  \label{kappadef}
\end{equation}
Thus
\begin{equation}
F= (u \cdot \kappa)\; \alpha \wedge l^\flat\;.
\label{mixedF}
\end{equation}
Clearly $F$ is invariant under transformations in eqs.(\ref{lntrans}) and (\ref{satrans1}) since
\begin{equation}
   \kappa \rightarrow \frac{\kappa}{f}\;. 
   \label{kappatrans}
\end{equation}
In order to understand the governing equations for both $u$ and $\kappa$, eq.(\ref{mixedF}) is our preferred form for $F$ rather than eq.(\ref{simF}). 
In a foliation adapted frame $(s, l, \alpha^\sharp, n)$, the metric takes the simple form given by
\begin{equation}
   g= \left(
    \begin{array}{cccc}
      1 & 0 &0&0 \\
       0 & 0 &0&-1 \\
      0 & 0&1&0 \\
      0&-1&0&0
    \end{array}
  \right)=g^{-1} \;. 
  \label{frameg}
\end{equation}
The above form of the metric will be useful as we take the Hodge-Star dual of forms. Having developed all the necessary background material we begin by formulating a geometric version of theorem \ref{existuninull}.
\begin{theorem} {\rm({\bf Version} 2)}
Let ${\cal F}$ be a null foliation of ${\cal M}$ with metric $g$. Let $\big(U_p, \phi_p= (x^1, \dots, x^4)\big)$ be a null foliation adapted chart about any arbitrary point $p \in {\cal M}$, and let $(s, l, \alpha^\sharp, n)$ be a null foliation adapted frame for ${\cal F}$ in $U_p$.  Then $F$ given by  eq.(\ref{mixedF}) for any smooth function $u$ on ${\cal F}|_{U_p}$ is a unique class of force-free null solutions in $U_p$  such that $\ker F$ contains exactly the vectors in $T({\cal F}_a)$ if and only if 

\begin{subequations}
\begin{eqnarray}
    d\alpha(l,\alpha^\sharp)=ds^\flat(l,s)\\
  dl^\flat(l, \alpha^\sharp)=0\;.
\end{eqnarray}
\label{cond}
\end{subequations}
\label{geoexistuninull}
\end{theorem}
\begin{proof}
$dF=0$ implies that
$$d (u\cdot\kappa) \wedge \alpha \wedge l^\flat=(u\cdot\kappa)\; \big[\alpha \wedge dl^\flat-d\alpha \wedge l^\flat\big]\;.$$
Eqs.(\ref{dflat}) and (\ref{flatal}) reduces the above equation to the form
\begin{equation}
   d (u\cdot\kappa) \wedge \alpha \wedge l^\flat=(u\cdot\kappa)\; \alpha \wedge l^\flat \wedge (A+D)\;. 
   \label{df=0cond}
\end{equation}
Meanwhile, from eq.(\ref{frameg}), we get that
$$*F=(u\cdot\kappa) * ( \alpha \wedge l^\flat) = (u\cdot\kappa) \;l^\flat\wedge s^\flat\;.$$
Then
\noindent\vskip0.2in
$J=d*F$
$$= d(u\cdot\kappa) \wedge l^\flat\wedge s^\flat +(u\cdot\kappa) \; dl^\flat\wedge s^\flat-(u\cdot\kappa) \; l^\flat\wedge ds^\flat\;.$$
Once again, using eq.(\ref{dflat}) and (\ref{flatal}), we get that
\noindent\vskip0.2in
$J= d(u\cdot\kappa) \;\wedge l^\flat\wedge s^\flat +(u\cdot\kappa) \; l^\flat \wedge (A_n \;n^\flat+A_\alpha \;\alpha) \wedge s^\flat$
\begin{equation}
 +(u\cdot\kappa) \;\alpha \wedge B_n\;n^\flat \wedge s^\flat -(u\cdot\kappa) \; l^\flat\wedge ds^\flat\;.  
 \label{Jexplicit}
\end{equation}
From above and eq.(\ref{FFwedgecond}), one of the two force-free conditions become
$$0=J \wedge l^\flat \iff B_n =0 \;.$$
Therefore
$$B_n=dl^\flat(l, \alpha^\sharp)=0\;.$$
Therefore eq.(\ref{cond}b) is a necessary condition. Finally, imposing $J \wedge \alpha =0$ we get that
$$0=d(u\cdot\kappa) \;\wedge l^\flat\wedge s^\flat \wedge \alpha +(u\cdot\kappa)\;A_n \; l^\flat \wedge n^\flat \wedge s^\flat \wedge \alpha$$
$$-(u\cdot\kappa) \; l^\flat\wedge ds^\flat \wedge \alpha\;$$
$$=\alpha \wedge l^\flat \wedge n^\flat \wedge s^\flat\;[(A_n +D_n)-A_n+ds^\flat _{ns}]\;. $$
Here $ds^\flat _{ns}$ denotes the $n^\flat \wedge s^\flat$ component of the 2-form $ds^\flat$. 
The final expression above was obtained by using eq.(\ref{df=0cond}) to eliminate the term containing $d(u \cdot \kappa)$. 
Therefore the only remaining requirement for FFE is given by
$$D_n = -ds^\flat _{ns}\;.$$
By definition,
\begin{equation}
   ds^\flat _{ns}=-ds^\flat(l,s)\;.
   \label{dsfalt}
\end{equation}
Also, from eq.(\ref{flatal}) it is clear that
\begin{equation}
    D_n=d\alpha (l, \alpha^\sharp)\;.
    \label{duns}
\end{equation}
This completes the proof of the theorem.
$\blacksquare$
\end{proof}
For the moment, it appears that the geometric version of the theorem is no more enlightening than the coordinate adapted version. But, we will reformulate eqs.(\ref{cond}) into physically meaningful terms shortly.
It is important to note that the expression $(u\cdot\kappa)$ is not present in the requirements of the theorem above. This is already consistent with what we know from \cite{Menon_FF20} since there will always be inherent freedom in the choice of $u(x^3, x^4)$ for null, force-free fields.
In the remainder of this section, we will interpret eqs. (\ref{cond}). The following theorem will show that $l$ is a pregeodesic tangent vector field. I.e., the integral curve of $l$ when suitably parametrized is a (null) geodesic. This is the key result that will allow for a geometric interpretation of eq. (\ref{cond}a).
\begin{theorem}
$ dl^\flat(l, \alpha^\sharp)=0$ if and only if $l$ is a pregeodesic vector field.
\label{pregeolem}
\end{theorem}
\begin{proof}Let $p \in {\cal F}_a$, be as in theorem {\rm \ref{geoexistuninull}}, and let $(s, l, \alpha^\sharp, n)$ be a null foliation adapted frame for ${\cal F}$ about some open set $U_p$, where $p \in {\cal F}_a \in {\cal F}$. Then
$$ dl^\flat(l, \alpha^\sharp)= (l^\mu \alpha^\nu- \alpha^\mu l^\nu)\;\nabla_\mu l_\nu \;.$$
Eq.(\ref{cond}b) implies that
$$ dl^\flat(l, \alpha^\sharp)= l^\mu \alpha^\nu\;\nabla_\mu l_\nu =0\;,$$
or
$$\alpha_\nu\;(\nabla_l\; l^\nu)=0\;.$$
Additionally, since
$$l_\nu\;(\nabla_l\; l^\nu)=0\;,$$
we get that
$$\nabla_l\; l^\nu \in \ker F\;.$$
Now let $\tilde s$ be the Lie transport of $s|_p$ along $l$ passing through the point $p$ in ${\cal F}_a$. I.e., $[\tilde s, l]=0$.
 Then,  since the torsion tensor is trivial in general relativity,
 \vskip0.2in \noindent
 $g(\nabla_l l,s)|_p$
 $$=g(\nabla_l l, \tilde s)|_p= -g(\nabla_l \tilde s, l)|_p=g(\nabla_{\tilde s} l, l)|_p=0\;.$$
 I.e., $\nabla_l l$ is either vanishing, or at most proportional to $l$. The converse of what we have proved is seen to be true by simply reversing the argument.
$\blacksquare$
\end{proof}
The null mean curvature, or the null expansion scalar, $\theta$, for the congruence generated by $l$ is given by
\begin{equation}
   \theta= \frac{1}{2} \;\Big[\;g(\nabla_s\; l, \;s)+g(\nabla_{\alpha^\sharp} \;l, \;\alpha^\sharp)\;\Big]\;. 
   \label{nmcdef}
\end{equation}
In addition to the transformation in eq.(\ref{satrans1}), eq.(\ref{nmcdef}) is invariant under orthogonal transformations involving $\alpha^\sharp$ and $s$:
$$\left(
           \begin{array}{c}
             \bar\alpha^\sharp \\
             \bar s \\
           \end{array}
         \right)= O\left(
           \begin{array}{c}
             \alpha^\sharp \\
             s\\
           \end{array}
         \right)\;.
$$
Here $O$ is any $2 \times 2$ spacetime dependent orthogonal matrix. Under eq.(\ref{lntrans}), clearly
$$\theta \rightarrow f \theta\;.$$
\begin{definition}Let $(s, l, \alpha^\sharp, n)$ be a null foliation adapted frame for a null foliation  ${\cal F}$. 
If
$$\theta= g(\nabla_s\; l, \;s)\;, $$
we say that ${\cal F}$ admits an equipartition of null mean curvature with respect to the null pregeodesic vector field $l$. 
\end{definition}
Note that when ${\cal F}$ admits an equipartition of null mean curvature with respect to the null pregeodesic vector field $l$ the value of $\theta$ is equally shared by $g(\nabla_s\; l, \;s)$ and $g(\nabla_{\alpha^\sharp}\; l, \alpha^\sharp\;)$.
\begin{lemma}Let $(s, l, \alpha^\sharp, n)$ be a null foliation adapted frame for a null foliation  ${\cal F}$. Then
$d\alpha(l,\alpha^\sharp)=ds^\flat(l,s)$ (i.e., eq.(\ref{cond}a) is true) if and only if ${\cal F}$ admits an equipartition of null mean curvature with respect to the null pregeodesic vector field $l$. 
\label{lemequi}
\end{lemma}
\begin{lproof}
$$ds^\flat(l,s)= (l^\mu s^\nu- s^\mu l^\nu)\;\nabla_\mu s_\nu= s_\nu \nabla_s \;l^\nu = g(\nabla_s\; l, \;s)\;.$$
Similarly,
$d\alpha(l,\alpha^\sharp)=g(\nabla_{\alpha^\sharp}\; l, \;\alpha^\sharp)\;.$ A direct substitution of the expressions above into eq.(\ref{nmcdef}) gives the needed result.
$\blacksquare$
\end{lproof}
We finally write down the conditions for a null force-free field in completely geometric terms.
\begin{theorem} {\rm({\bf Version} 3)}
Let ${\cal F}$ be a null foliation of ${\cal M}$ with metric $g$. 
 Let $\big(U_p, \phi_p= (x^1, \dots, x^4)\big)$ be a null foliation adapted chart about any arbitrary point $p \in {\cal M}$, and let $(s, l, \alpha^\sharp, n)$ be a null foliation adapted frame for ${\cal F}$ in $U_p$. Then $F$ given by  eq.(\ref{mixedF}) for any smooth function $u$ on ${\cal F}|_{U_p}$ is a unique class of null, force-free  solutions in $U_p$ such that $\ker F$ contains exactly the vectors in $T({\cal F}_a)$ if and only if 
\begin{itemize}
    \item $l$ is a pregeodesic vector field, and
    \item ${\cal F}$ admits an equipartition of null mean curvature with respect to the null pregeodesic vector field $l$. 
\end{itemize}
\label{existunifianl}
\end{theorem}
\begin{proof}Suppose we picked a different $\tilde s \in \ker F$, then $\tilde s = s + g\;l$ for some smooth function $g$. Clearly, as mentioned before
$$g(\nabla_{\tilde s}\; l, \;\tilde s)= g(\nabla_s\; l, \;s)\;.$$
The rest follows from theorems \ref{geoexistuninull}, \ref{pregeolem} and lemma \ref{lemequi}.
$\blacksquare$
\end{proof}
 As is evident, the conditions of the above theorem are impervious to the redundancies in the choice of a null foliation adapted frame as given in eq. (\ref{lntrans}), and (\ref{satrans1}). It is now easy to pin down the conditions a  null foliation of ${\cal M}$ must satisfy to permit the possible existence of a null, force-free field.
\begin{theorem}
Let $F$ be a null and force-free electromagnetic field on ${\cal M}$. Let ${\cal F} = \cup\; {\cal F}_a$ be the associated null foliation of ${\cal M}$ such that $\ker F$ contains exactly the vectors in $T({\cal F}_a)$. Then through each point point of  ${\cal M}$, there exists a null geodesic vector field $l \in T({\cal F}_a)$, and ${\cal F}$ admits an equipartition of null mean curvature with respect to the null geodesic vector field $l$. 
\label{anotherv}
\end{theorem}
\begin{proof}
The above conditions do not refer to charts, and yet these conditions, as per the previous theorem, have to hold locally, and hence globally.
$\blacksquare$
\end{proof}
\begin{definition}
A null foliation ${\cal F} = \cup\; {\cal F}_a$  is a null field sheet foliation if there exists a null geodesic congruence $l$ in ${\cal M}$ such that $l\in T({\cal F}_a)$ and ${\cal F}$ admits an equipartition of null mean curvature with respect to the null geodesic vector field $l$.
\end{definition}
Then, a null field sheet foliations will always permit local null and force-free solutions.
\section{Dual Solutions and Generalized Dual Solutions}
\begin{theorem}
Let ${\cal F}$ be a null field sheet foliation of ${\cal M}$, and let $(s, l, \alpha^\sharp, n)$ be a null foliation adapted frame for ${\cal F}$. Suppose that the pair of vector fields $\alpha^\sharp$ and $l$ forms an involutive distribution, then there exists a new class of local ``dual " null force-free solutions $\tilde F$ such that the kernel of $\tilde F$ is exactly the span of $\alpha^\sharp$ and $l$. 
\label{dualsol}
\end{theorem}
\begin{proof}
If $l$ and $\alpha^\sharp$ form an involutive distribution, then in theorem \ref{existunifianl}, we simply apply the following substitution:
$$(s, l, \alpha^\sharp, n) \rightarrow (\alpha^\sharp, l, s, n)\;,$$
since 
$$\theta= g(\nabla_s\; l, \;s)= g(\nabla_{\alpha^\sharp}\; l, \;\alpha^\sharp)\;.$$
There is a subtle point hidden here: if $l$ and $\alpha^\sharp$ do not form an involutive distribution, we will not have that $d\tilde F=0$.
$\blacksquare$
\end{proof}
It is now reasonable to ask if the span of $l$ and any linear combination of $s$ and $\alpha^\sharp$ can generate a null field sheet foliation. The required condition for a generalized class of null, force-free dual solutions are easy to write down.
\begin{definition}
Let ${\cal F}$ be a null field sheet foliation of ${\cal M}$, and let $(s, l, \alpha^\sharp, n)$ be a foliation adapted frame for ${\cal F}$. Then  $l$ admits a uniform equipartition of null mean curvature if
$$g(\nabla_s l, \alpha^\sharp)+g(\nabla_{\alpha^\sharp} l, s) =0\;. $$
\end{definition}
Note the distinction that while ${\cal F}$ may admit an equipartition of null mean curvature with respect to $l$, it is $l$ itself that admits a uniform equipartition of null mean curvature.
\begin{theorem}
Let ${\cal F}$ be a null field sheet foliation of ${\cal M}$, and let $(s, l, \alpha^\sharp, n)$ be a null foliation adapted frame for ${\cal F}$. Let $l$  admit a uniform equipartition of null mean curvature, and for some smooth functions $A$ and $B$, let
$$ \hat s = A\; s+ B\;\alpha^\sharp$$ be a unit vector field such that the span of $l$ and $\hat s$ form an integrable distribution and thus generate submanifolds  that form a foliation $\hat {\cal F}$ of ${\cal M}$. Then  $\hat {\cal F}$ is a null field sheet foliation.
\label{hatsoltheo}
\end{theorem}
\begin{proof}As before, the requirement that the span of $l$ and $\hat s$ form an integrable distribution ensures that $d\hat F=0$.
Clearly  $g(l, \hat s)=0$, and
$$1=g(\hat s, \hat s) = A^2 + B^2\;.$$
Then
$$g(\nabla_{\hat s} l, \hat s) = A^2\; g(\nabla_s l, s) + B^2\;  g(\nabla_{\alpha^\sharp} l, \alpha^\sharp) $$
$$+ AB\; \big[\;g(\nabla_s l, \alpha^\sharp)+g(\nabla_{\alpha^\sharp} l, s)  \big]$$
$$= (A^2 + B^2)\;\theta = \theta\;.$$
I.e., since $l$ is a null pregeodesic vector field, and $\hat {\cal F}$ admits an equipartition of null mean curvature, theorem \ref{existunifianl} gives us the necessary result. 
$\blacksquare$
\end{proof}
The resulting solution $\hat F$ will be of the form
\begin{equation}
    \hat F = (u \cdot \hat \kappa) \;\hat \alpha \wedge l^\flat\;,
    \label{hatsol}
\end{equation}
where
\begin{equation}
  \left(
           \begin{array}{c}
             \hat s \\
             \hat \alpha^\sharp \\
           \end{array}
         \right)=\left(
    \begin{array}{cc}
      A & B \\
      -B & A \\
    \end{array}
  \right)\left(
           \begin{array}{c}
             s \\
             \alpha^\sharp \\
           \end{array}
         \right)\;
         \label{lrot}
\end{equation}
is clearly nothing more than a spacetime dependent point wise rotation about $l$.
\section{Consistency in Formalism}
It would appear that theorems \ref{existuninull} and \ref{existunifianl} state the requirements for the existence of a null and force-free electromagnetic field, seemingly, in two different ways. 
In this section, we will show that both sets of requirements are in fact equivalent. We begin by writing down the equations of FFE, in the null field case, in a standard adapted chart.
\begin{lemma} In a null foliation adapted chart, the force-free condition, $J \wedge dx^3 = 0= J \wedge dx^4$,
if and only if
\begin{subequations}
\begin{eqnarray}
   l(\ln |\kappa|) + dl^\flat (n,l)+ ds^\flat(l,s)=0\\
  dl^\flat(l, \alpha^\sharp)=0\;.
\end{eqnarray}
\label{chartverse}
\end{subequations}
\label{lemcharFF}
\end{lemma}
Notice that the second requirement in the above lemma is the same as eq.(\ref{cond}b).
\begin{lproof}
The form of $J$ we will be working with is given by eq.(\ref{Jexplicit}). Consider the first term in the expression for $J$ given by 
$$d(u\cdot\kappa) \;\wedge l^\flat\wedge s^\flat\;.$$
To impose the force-free condition we have to take the wedge product of the above expression with $dx^3$ (and similarly with $dx^4$).
$$d(u\cdot\kappa) \;\wedge l^\flat\wedge s^\flat \wedge dx^3= l^\flat _4\; d(u\cdot\kappa) \;\wedge dx^4 \wedge s^\flat \wedge dx^3\;.$$
Here the factor $l^\flat _4$ comes from eq.(\ref{chart2frame}). Since $dx^3$ and $dx^4$ span the same plane as $\alpha$ and $l^\flat$, the only relevant component of $d(u\cdot\kappa)$ in the above expression is along the $n^\flat$ direction. I.e.,
$$d(u\cdot\kappa) \;\wedge l^\flat\wedge s^\flat \wedge dx^3=-l^\flat _4\;\;d(u\cdot\kappa)(l)  \;n^\flat \wedge dx^4 \wedge s^\flat \wedge dx^3\;.$$
Proceeding in a similar manner we find that
$$0=J \wedge dx^3 = n^\flat \wedge s^\flat \wedge dx^4 \wedge dx^3\;\; \Big(l^\flat _4\;\Big[l(u\cdot\kappa)$$
$$ + (u\cdot\kappa)\; dl^\flat (n,l)+ (u\cdot\kappa)\; ds^\flat(l,s)\Big]+\alpha_4\;(u\cdot\kappa)\; dl^\flat(l, \alpha^\sharp)\Big)\;.$$
Since $l(u)=0$ {\rm (}$l$ is in $\ker F$, while $u=u(x^3, x^4)${\rm )}, the above equation reduces to
$$l^\flat _4\;\big[l(\ln |\kappa|) +  dl^\flat (n,l)+  ds^\flat(l,s)\Big] +\alpha_4\; dl^\flat(l, \alpha^\sharp)=0\;.$$
In exactly the same way, $J \wedge dx^4 =0$ implies that
$$l^\flat _3\;\big[l(\ln |\kappa|) +  dl^\flat (n,l)+  ds^\flat(l,s)\Big] +\alpha_3\; dl^\flat(l, \alpha^\sharp)=0\;.$$
The last two equations hold true if and only if the conditions of the lemma hold true.
$\blacksquare$
\end{lproof}
From eq.(\ref{chartverse}a) we see that while $\kappa$, the factor that allows transition from an adapted chart to the adapted frame, is determined by chosen frame, $u$ continues to be a free function. This feature is never lost in the null case.

So far, the foliation adapted frame has been very abstract. We will now make a specific choice using our standard adapted chart. Since $M^4 = 0$ in the null case, and $M^3$ is always identically zero, elevate $M^r$ to a vector by defining $M\in T(U_p)$ by
$M=M^\mu \partial_\mu = M^1 \partial_1 + M^2 \partial_2$, where $\{\partial_\mu\}$ refer to the tangent bases vectors in the standard adapted chart.
Let 
\begin{equation}
 \chi_M = g^{34} \;dx^3-g^{33} \;dx^4\;.
 \label{chimM}
\end{equation}
A simple calculation reveals that $\chi_M ^\sharp = M$, and so $\chi_M$ is a null dual vector whose associated tangent vector is a null vector in $\ker F$. I.e., $M$ is a good candidate for $l$ in a foliation adapted frame.
In the remainder of this section, we will specifically work with $M$ as the choice for $l$. I.e., the null foliation adapted frame takes the form $(s,M, \alpha^\sharp, n)$. 
Eq. (\ref{divfreeFFeq}) implies that, for $M$, $*d*M^\flat =0$.
\begin{lemma}
In the specific null foliation adapted frame
$(s,M, \alpha^\sharp, n)$, $*d*M^\flat =0$ if and only if
\begin{equation}
    dM^\flat(n,M)+d\alpha(M,\alpha^\sharp)+ds^\flat(M,s)=0\;. 
    \label{divmvanish}
\end{equation}
\end{lemma}
\begin{lproof}
From eq.(\ref{frameg}), we get that
$*M^\flat=M^\flat \wedge \alpha \wedge s^\flat$.
Then  from eq.(\ref{dflat}) and (\ref{flatal})
$$d*M^\flat= M^\flat\wedge n^\flat  \wedge \alpha \wedge s^\flat \;[A_n+D_n-ds^\flat _{ns}]=0\;.$$
Eq.(\ref{dflat}) implies that $A_n=dM^\flat(n,M)$. $ds^\flat _{ns}$ and $D_n$ are given by eqs.(\ref{dsfalt}) and (\ref{duns}). Inserting the expressions for the individual terms in the equation above gives us the needed conclusion.
$\blacksquare$
\end{lproof}
We have already picked a fixed choice for $l$ for the foliation adapted frame. Now we settle on a particular choice of $\alpha$. Assuming that $g^{33} \neq 0$, set 
\begin{equation}
    \alpha = \frac{-1}{\sqrt{g^{33}}} dx^3\;.
    \label{alphadef}
\end{equation}
If $g^{33} = 0$, we can always use the unit 1-form proportional to $dx^4$, or simply relabel the coordinates by $x^3 \leftrightarrow x^4$ since both $g^{33}$ and $g^{44}$ cannot be $0$ at the same time.
The ``-" sign is included only to make $\kappa >0$.
\begin{lemma}
With the above choice of $\alpha$, eq.(\ref{divmvanish}) is true if and only if eq.(\ref{chartverse}a) is true.
\end{lemma}
\begin{lproof}
All that remains to be shown is that
$$M(\ln |\kappa|)=d\alpha(M, \alpha^\sharp)\;.$$
From eqs.(\ref{kappadef}), (\ref{alphadef}) and our choice of $l$, we get that
$$\kappa = \frac{1}{\sqrt{g^{33}}}\;.$$
Then
$$M(\ln|\kappa|)= \frac{-1}{\sqrt{g^{33}}} M\left(\sqrt{g^{33}}\right)\;.$$
On the other hand,
$$d\alpha(M, \alpha^\sharp)=M^\mu \alpha^\nu \;(\partial_\mu \alpha_\nu) - \alpha^\mu M^\nu\;(\partial_\mu \alpha_\nu)\;.$$
But $\alpha_\nu=0$ whenever $\nu \neq 3$, and $M^\nu =0$ when $\nu =3$. Therefore,
$$d\alpha(M, \alpha^\sharp)=M^\mu \alpha^\nu \;(\partial_\mu \alpha_\nu)=\alpha^3 M(\alpha_3)$$
$$=\frac{-1}{\sqrt{g^{33}}} M\left(\sqrt{g^{33}}\right)= M(\ln|\kappa|)\;.$$
$\blacksquare$
\end{lproof}
Since $N^r \propto M^r$, we will elevate $N^r$ to the status of a vector field as well.
Now we have to seek out the role played by the vanishing of the divergence of $N^r$. A casual examination reveals that when $g^{44} \neq 0$,
$$N= \frac{g^{44}}{g^{34}} \;M\;.$$
Note that if $g^{44} \neq 0$, we must necessarily have that $g^{34} \neq 0$. Since $M$ is divergence free,
\begin{equation}
  \nabla_\mu N^\mu =0 \iff M\left(\frac{g^{44}}{g^{34}}\right)=0\;.  
  \label{chartdivn=0}
\end{equation}
If $g^{44} = 0$, then $g^{34} =0$, and $N=0$.
\begin{lemma} When $l=M$, and $\alpha$ is as given by eq.(\ref{alphadef}),
$$ \nabla_\mu N^\mu=0$$ if and only if eq.(\ref{chartverse}b) is true.
\label{geolem}
\end{lemma}
\begin{lproof} The result follows from a straightforward calculation. First, let $g^{44} \neq 0$.
$$dM^\flat (\alpha^\sharp, M)= -\alpha^\nu M(M_\nu)$$
$$=\frac{g^{33}}{\sqrt{g^{33}}} M(M_3)+\frac{g^{34}}{\sqrt{g^{33}}} M(M_4)\;.$$
Or,
$$g^{33} M(g^{34})-g^{34} M(g^{33})=- (g^{34})^2\;M\left(\frac{g^{33}}{g^{34}}\right)\;.$$ 
Eq.(\ref{chartdivn=0}) now gives us the needed result when $g^{44} \neq 0$.
On the other hand if $g^{44} = 0$, clearly $N$ is trivially divergence free and 
$$dM^\flat (\alpha^\sharp, M)=\frac{g^{33}}{\sqrt{g^{33}}} M(M_3)\;.$$
From eq.(\ref{chimM}), since $\chi_M = M^\flat$, we get that $M_3 =0$. Once again the lemma holds true.
$\blacksquare$
\end{lproof}
\begin{theorem}
Theorems \ref{existuninull}, \ref{geoexistuninull}, and \ref{existunifianl} are equivalent.
\end{theorem}
\begin{proof}
The result follows from lemmas \ref {lemcharFF} to \ref{geolem}.
$\blacksquare$
\end{proof}
When using a new adapted chart, we have that
$$\kappa \rightarrow \frac{\kappa}{D}\;,$$
where $D$ is as defined in eq. (\ref{ddef}).
Therefore, from eq. (\ref{utrans}) we get that
\begin{equation}
   u  \cdot \kappa \rightarrow (D \cdot u) \cdot \frac{\kappa}{D}= u  \cdot \kappa \;. 
   \label{ucapinv}
\end{equation}
\begin{theorem}
Eqs. (\ref{chartverse} a, and b) in lemma \ref{lemcharFF} do not depend on the adapted chart used nor in the freedom of choice in $s, l,\alpha$ and $n$ given by eqs.(\ref{lntrans}) and (\ref{satrans1}).
\end{theorem}
\begin{proof}
Eq. (\ref{ucapinv}) clearly implies that
eqs.(\ref{chartverse}a) and (\ref{chartverse}b) remain true under an adapted chart transformation.
Now consider the transformation in eq.(\ref{lntrans}). Then from eq.(\ref{kappatrans}) we see that eq.(\ref{chartverse}a) becomes
$$f\;l \ln\left(\frac{\kappa}{f}\right) + d (f\;l^\flat) \left(\frac{n}{f}, f\;l\right) + ds^\flat(f\;l,s)$$
$$=f \;\left[l \ln \kappa + d l^\flat (n, l) + ds^\flat(l,s)\right]-f\;\ln f + df \wedge l^\flat\; (n,l)$$
$=-f\;\ln f +l(f)=0\;.$

\vskip0.2in\noindent I.e., eq. (\ref{chartverse}a) holds true under the transformation given in  eq. (\ref{lntrans}). Since eq. (\ref{chartverse}a) does not depend on $\alpha$ consider the transformation in eq.(\ref{satrans1}) when $h=0$. Then,
\vskip0.2in \noindent 
$l \ln \kappa + d l^\flat \left(n+\frac{g^2}{2}\;l+g\;s, l\right) + d(s^\flat+g\;l^\flat)(l,s + g\;l)$
\vskip0.2in \noindent 
$=l \ln \kappa + d l^\flat (n, l) + ds^\flat(l,s)+dg\wedge l^\flat (l,s) +g\;dl^\flat(l,s)=0$\vskip0.2in\noindent Eq.(\ref{dflat}) was used in obtaining the final equality above. A similar calculation reveals that eq.(\ref{chartverse}b) is not affected by picking a different adapted chart or by transformations in   eqs.(\ref{lntrans}) and (\ref{satrans1}).
$\blacksquare$
\end{proof}
\section{Currents in a Null and Force-Free Field}
Note that
$$* \;l^\flat \wedge \alpha \wedge n^\flat = s^\flat\;,$$
and
$$*\; s^\flat \wedge l^\flat  \wedge \alpha  = l^\flat.$$
Then, from eq.(\ref{Jexplicit}) and (\ref{chartverse} a), it is straight forward to write down the expression for $j$. The current density vector $j$ of a null force-free field in a null foliation adapted frame is given by
\vskip0.2in\noindent
$j=-(u \cdot \kappa)\; \times$
\begin{equation}
  \big[\alpha^\sharp(\ln u \cdot \kappa)+\;dl^\flat(n,\alpha)+ ds^\flat(\alpha^\sharp, s)\big] l^\flat+\;ds^\flat(l,\alpha^\sharp)\;s^\flat\;. 
  \label{jfinal}
\end{equation}
In general, a null, force-free field does not require the current to flow along the null geodesics of the foliation. 
\begin{theorem}
Let ${\cal F}$ be a null field sheet foliation of ${\cal M}$, and let $(s, l, \alpha^\sharp, n)$ be a null foliation adapted frame for ${\cal F}$. Let $F$ be a particular solution in the adapted frame. Then
\begin{itemize}
    \item the current density vector $j$ for the null, force-free field $F$  is along the null pregeodesic $l$ if and only if 
\begin{equation}
  ds^\flat (l, \alpha^\sharp)=0\;.  
  \label{geocur}
\end{equation}
\item if further
\begin{equation}
   \alpha^\sharp(\ln u \cdot \kappa)+\;dl^\flat(n,\alpha)+ ds^\flat(\alpha^\sharp, s)=0\;, \label{geocomp}
\end{equation} then $F$ describes a vacuum solution.
\end{itemize}
 \label{currentcomp}
\end{theorem}
\begin{proof}
From eq. (\ref{jfinal}) we see that the requirements of the theorem are necessary.
Eq.(\ref{ucapinv}) once again shows that the same equations hold under an adapted chart transformation.
Under a transformation given by eq.(\ref{satrans1}), eq.(\ref{geocur}) becomes
\vskip0.2in \noindent
$d(s+g\;l)^\flat (l, \alpha^\sharp+ h\;l)$
$$=ds^\flat(l,\alpha^\sharp)+g\;dl^\flat(l, \alpha^\sharp) + dg \wedge l^\flat(l, \alpha^\sharp)=0\;.$$
The final equality above was obtained using eqs.(\ref{geocur}) and (\ref{cond}b). Clearly eq. (\ref{geocur}) is also true under the transformation given by eq. (\ref{lntrans}).
\vskip0.2in
To show that eq. (\ref{geocomp}) remains true under an adapted frame transformation consider first the case in eq. (\ref{satrans1}) when $g=0$. In this case
$$(\alpha^\sharp +h \;l)\;(\ln|u \cdot \kappa|) + dl^\flat\left(n+\frac{h^2}{2}\;l + h\;\alpha^\sharp,\; \alpha^\sharp + h\;l\right)$$
$$+\;ds^\flat(\alpha^\sharp + h\;l,s)$$
$$=(1+h)\;\big[l(\ln |\kappa|) + dl^\flat (n,l)+ ds^\flat(l,s)\big]+\frac{h^2}{2}dl^\flat (\alpha^\sharp, l)=0\;.$$
The final equality above was obtained using eq.(\ref{chartverse}a) and eq.(\ref{cond}b).
The other remaining cases can be shown in similar manner, and we omit the details for brevity.
$\blacksquare$
\end{proof}

\section{Examples}

\subsection*{A Null Force-Free Field in FLRW Cosmology}
 For concreteness, we set the sectional curvature of spacetime, $k = +1$. In the hyperspherical coordinate system $(t,r,\theta,\varphi)$,
 $$g = -dt^2 + a^2(t) \big[dr^2+\sin^2 r\; d\Omega^2\big]\;.$$
Here $d\Omega^2$ is the standard metric on the unit $2$-sphere.
Our calculation here does not depend on the explicit choice for $a(t)$, and so the matter content of the universe remains unfixed.  Let
$$l= a(t)\; \partial_t +\partial_r\;.$$ Then $l$ is a null vector field, and moreover
$$\nabla_l l = 2 \dot a\; l\;.$$
Here $\dot a$ refers to derivative of $a(t)$ with respect to our time coordinate $t$. I.e., $l$ is a pregeodesic congruence. Clearly,
$$l^\perp = {\rm span}\;\;\{l, \partial_\theta, \partial_\varphi\}\;.$$
 Let
 $$s= \frac{\partial_\theta}{a \sin r}\;, \;\;{\rm and}\;\; \alpha^\sharp = \frac{\partial_\varphi}{ a \sin r \sin\theta}\;.$$
 By construction, we are looking for a null force-free field $F$, whose kernel consists of the span of $l$ and $s$. 
Since
 $$g(\nabla_s l,s)= \;\cot r + \dot a\; =g(\nabla_{\alpha^\sharp} l,\alpha^\sharp)\;,$$
 foliations generated by submanifolds whose tangent spaces are spanned by $l$ and $s$ admits an equipartition of null mean curvature. Consequently, theorem \ref{anotherv} guarantees the existence of a null force-free field. To identify the variables $x^3$ and $x^4$ in $u$, we begin by constructing a foliation adapted chart. To this end define commuting vector fields
 $$X_1 = l, \; X_2 = \partial_\theta, \; X_3 = \partial_\varphi, \;{\rm and}\;\;X_4 = \partial_r\;.$$
 Let $\{x^i\}$ be a chart such that $\partial_{x^i} = X_i$. Then
 $$dx^1= dt/a, \;dx^2 = d\theta, \;dx^3 = d\varphi\;,$$
 and
 $$dx^4 = -\frac{dt}{a} + dr\;.$$
 Here the force-free null field takes the form (eq.(\ref{finalformF}))
 $$F=u(\varphi, x^4) \left(\frac{-dt}{a} +dr\right) \wedge d\varphi\;,$$
 where $u(\varphi, x^4)= u(t,r,\varphi)$ is such that
 $$X_1(u) = a(t)\; \partial_t \;u + \partial_r\; u =0\;,$$
 and trivially
 $$X_2 (u)= \partial_\theta u =0\;.$$
 It will always be the case that the only remaining constraint $X_1(u)=0=X_2(u)$ is what allows for 2 parameters of freedom in choosing $u$. These are the class of solutions we have been referring to in the null, force-free case.
 Here $ds^\flat (l, \alpha^\sharp)=0$, and so from theorem \ref{currentcomp}, the current density vector is along $l$. A direct computation gives that
 $$j=\frac{u_{,\varphi}}{a^4 \sin^2 r \sin \theta}\;l\;,$$
 and indeed, as expect the solution is force-free.

\subsection*{The Generalized Dual Solutions Generated by $l= a(t)\; \partial_t +\partial_r$ in FLRW Spacetime}
It is easy to verify that
 $$g(\nabla_{\alpha^\sharp} \;l, s ) = 0 = g(\nabla_s \;l, \alpha^\sharp)\;.$$
 I.e., $l$ admits a uniform equipartition of null mean curvature, and so there is a possibility of further generalizing the previous solution. All that remains is the search for involutive distributions containing $l$.
 Define
 $$\hat s = A(t,r,\theta,\varphi)\; s+ B(t,r,\theta,\varphi)\;\alpha^\sharp\;.$$
 Then $l$ and $\hat s$ is involutive if and only if
 \begin{equation}
    A\; l(B)=B \;l(A)\;. 
    \label{invcons}
 \end{equation}
When the above condition holds, theorem \ref{hatsoltheo} guarantees a generalized solution of the form
$$\hat F = (u \cdot \hat \kappa) \;\hat\alpha \wedge l^\flat\;,$$ where from eq.(\ref{lrot}) we see that
 $$\hat \alpha = - B\; s^\flat + A \;\alpha= -aB\sin r\;d\theta + aA \sin r \sin \theta \;d\varphi\;,$$
 and
 $$l^\flat = -a\; dt + a^2 \;dr\;.$$
 Therefore, there must be solutions of the type
 $$\hat F = a \sin r \;(u \cdot \hat \kappa) \;(-B\;d\theta + A \sin \theta \; d\varphi) \wedge l^\flat\;.$$
 Since $l(\sin \theta) =0$, we can rewrite $A \sin \theta$ as $A$.
 It is not easy to write down the adapted chart in this case. Consequently, there is no real way to guess the form of $\hat \kappa$ or even the free variables $x^3$ and $x^4$ in $u$. Thankfully, this is not needed. We can subsume the factor $a \sin r \;(u \cdot \hat \kappa)$ into $A$ and $B$. Of course, a priori, there is no longer the need for $A$ and $B$ to satisfy eq.(\ref{invcons}). With the reassignments, we get that
 \begin{equation}
     \hat F = (-B\;d\theta + A \; d\varphi) \wedge l^\flat\;.
     \label{Fcos}
 \end{equation}
  Now we simply have to enforce Maxwell's equations. Noting that
  $$dl^\flat = 2 a \dot a\; dt \wedge dr\;, $$
  and taking the exterior derivative of $\hat F$ we get that
  $$0=d\hat F = dt \wedge dr \wedge d\theta \;(a^2 B_{,t} + a B_{,r} + 2 a \dot aB)$$
    $$- dt \wedge dr \wedge d\varphi \;(a^2 A_{,t} + a A_{,r} + 2 a \dot aA)$$
     $$- a \;dt \wedge d\theta \wedge d\varphi \;(B_{,\varphi} +  A_{,\theta})+a^2 \;dr \wedge d\theta \wedge d\varphi \;( B_{,\varphi} +  A_{,\theta})\;.$$
     This imposes the conditions that
     \begin{equation}
       l(A) = -2 \dot a A,\;{\rm and}\;\;l(B) = -2 \dot a B 
       \label{cosconst1}
     \end{equation}
     which is consistent with eq. (\ref{invcons}) and a new constraint that
     \begin{equation}
        B_{,\varphi} +  A_{,\theta} =0\;.
        \label{cosconst2}
     \end{equation}
     In a similar manner, enforcing the in-homogeneous Maxwell equation, we get that
     $$j  =\left [(B \sin \theta)_{,\theta} - \frac{A_{,\varphi}}{\sin \theta}\right] \;\frac{1}{a^2 \sin^2 r \sin \theta}\;l\;.$$
     I.e., we get that the solution given by eq.(\ref{Fcos}) is a force-free null field in a FLRW spacetime when $k=+1$, when $A$ and $B$ satisfy constraints eq.(\ref{cosconst1}) and (\ref{cosconst2}). Also, there exists no further generalizations where the kernel of the null field contains the vector field $l= a(t)\; \partial_t +\partial_r$.
     
\subsection*{The $F_{ls}$ solution in Outgoing Kerr Background}
Consider as a second example the outgoing principal null geodesic congruence of the Kerr geometry given by
$$l=\partial_{r^\star}\;.$$
The underlying coordinate functions are that of the outgoing Kerr-Schild spacetime denoted by $(t^\star , r^\star , \theta^\star,\varphi^\star)$. 
They are related to the Boyer-Lindquist coordinates
by the following relations:
$$r^\star  = r\;,\;\;\;\; \theta^\star  = \theta\;,$$

$$d t^\star = dt -
\frac{r^2+a^2}{\Delta}dr\;, \;\;\; {\rm and} \;\;\;\; d \varphi^\star  =
d\varphi - \frac{a}{\Delta}dr.$$
The ``$\star$" is placed on $r$ and $\theta$ so that no
confusions arise while performing coordinate transformations.
In $ K^\star$, the spacetime metric is given by
\begin{equation}
g^\star_{ \mu  \nu} = \left[\begin{array}{cccc}
z-1& -1& 0& -za\sin^2(\theta)\\
-1& 0& 0& a\sin^2(\theta)\\
0 & 0& \rho^2 & 0\\
-za\sin^2(\theta)& a\sin^2(\theta)& 0& \Sigma^2 \sin^2(\theta)/\rho^2\\
\end{array}\right] \;.
\label{kerrstar}
\end{equation}
Here,
$$z=\frac{2Mr}{\rho^2}\;,$$
$$\rho^2 = r^2 + a^2
\cos^2\theta\;,\;\;\;\Delta = r^2 -2 M r + a^2\;,$$
$$
\Sigma^2 = (r^2 + a^2)^2 -\Delta \; a^2 \sin^2\theta\;,
$$
and
$$\sqrt {- g^\star} = \rho^2 \sin\theta.$$
The time orientation in $K^\star$ is given by the null congruence $\partial_{r^\star}$ which is set to be future pointing.
Also, since $\theta=\theta^\star$, $d\theta = d\theta^\star$ and $\partial_\theta = \partial_{\theta^\star}$, we will not make a distinction between the two coordinates in what follows. It is easy to verify that in this case
$$l^\perp={\rm span}\; \{l, s, \alpha^\sharp\}\;,$$
where $s$ and $\alpha^\sharp$ are unit spacelike vectors given by
$$s=\frac{1}{\sin (\theta) \sqrt{\rho^2}}\;\big[a \sin^2(\theta)\; \partial_{t^\star} +  \partial_{\varphi^\star}\big]\;,$$
and
$$\alpha^\sharp = \frac{1}{ \sqrt{\rho^2}}\;\partial_\theta\;.$$

We will begin by re-deriving the generalization to  the solution first derived in \cite{MD07}. This generalization was subsequently found by \cite{BGJ13}. However we proceed in a manner consistent with our geometric formulation of the theory of null, force-free electromagnetic fields. In this case, the kernel of our solution denoted by $F_{ls}$ is given by 
$$\ker F_{ls} = {\rm span}\;\{l, s\}\;.$$
Since $l$ is a null pregeodesic, all that remains is to check whether this foliation admits an equipartition of null mean curvature.
A routine calculation shows that the null mean curvature (unfortunately, also denoted by $\theta$) is given by
$$\theta= g(\nabla_s\; l, \;s)=g(\nabla_{\alpha^\sharp}\; l, \;\alpha^\sharp) = \frac{r}{\rho^2}\;.$$
All the requirements of theorem \ref{anotherv} have been satisfied, so already a solution is guaranteed. It is easy to write the explicit solution in the foliation adapted chart in the form given by eq.(\ref{finalformF}). To this end, define vector fields
$$X_1 = \partial_{r^\star},\; X_2= \sin (\theta) \sqrt{\rho^2} \;s\;,$$
$$X_3= \partial_{\varphi^\star}\;,\;{\rm and}\; X_4=\partial_\theta + 2t^\star \cot \theta \;\partial_{t^\star}\;.$$ 
It is easily verified that $[X_i, X_j]=0$ for $i,j=1,\dots,4\;.$ Therefore, there exists a foliation (determined by the distribution spanned by $l$ and $s$) adapted coordinate system $(x^1, \dots, x^4)$ such that
$$\frac{\partial}{\partial x^i}=X_i\;.$$
The coordinate $1$-forms transforms as
$$ dx^1= dr, \;dx^2=\frac{1}{a \sin^2 \theta} \big[dt^\star -2t^\star \cot\theta \;d\theta\big]\;,$$
$$dx^3=-dx^2 +d\varphi^\star\;,\;\; {\rm and}\;\; dx^4=d\theta\;.$$
In eq.(\ref{finalformF}),
$u$ does not depend on $x^1$ and $x^2$. This condition is encoded in the statement $X_1 (u)=0=X_2(u)$, and in our case reduces to the conditions
$$X_1(u)= \partial_{r^\star} u =0\;,$$
and
$$X_2 ( u) = 0 = a \sin^2(\theta)\; \partial_{t^\star} u +  \partial_{\varphi^\star} u\;.$$
I.e., we can write the null field as
\begin{equation}
   F_{ls} = \frac{u(x^3, \theta)}{a \sin^2 \theta}\;(-dt^\star + a \sin^2 \theta \; d\varphi^\star) \wedge d\theta\;. 
   \label{Flsv1}
\end{equation}
Note that
$$l^\flat =-dt^\star + a \sin^2 \theta \; d\varphi^\star\;, \;{\rm and}\;\; \alpha = \sqrt{\rho^2}\; d\theta\;.$$
Here
$$
  \left(
           \begin{array}{c}
             \alpha \\
             l^\flat \\
           \end{array}
         \right)=\left(
    \begin{array}{cc}
      0 & \sqrt{\rho^2} \\
      a \sin^2\theta & -2t\cot\theta  \\
    \end{array}
  \right)\left(
           \begin{array}{c}
             dx^3 \\
             dx^4 \\
           \end{array}
         \right)\;, 
$$
and so
$$
  \kappa=-(a \sin^2\theta\sqrt{\rho^2})^{-1}\;.
  $$
Therefore, from eq.(\ref{mixedF})
\begin{equation}
  F_{ls}= -\frac{u(x^3, \theta)\;( \rho^2)^{3/2}}{a \sin^2 \theta }\; \alpha \wedge l^\flat\;.  
   \label{Flsv2}
\end{equation}
Clearly the substitution  $-u/  a \sin^2\theta \rightarrow u\;,$ allows us to write the solution in a simpler way, but we will refrain from doing so in order to freely use the expressions previously derived in this paper. 
To compute the current density vector, we first trivially observe that
$$ds^\flat (l, \alpha^\sharp) = \alpha^\nu l(s_\nu) - l^\nu \alpha^\sharp(s_\nu)=0$$
since when $\alpha^\nu \neq 0$ or $l^\nu \neq 0$ we have that $s_\nu =0$. Therefore, from theorem \ref{currentcomp} we must have that $j$ is along $l$. To compute the component of $j$ along $l$ we separately evaluate the remaining terms in eq.(\ref{jfinal}). 
$$\alpha^\sharp(u \cdot \kappa)= \frac{1}{\sqrt{\rho^2}}\;\partial_\theta \left(\frac{-u}{a \sin^2 \theta \sqrt{\rho^2}}\right)$$
$$=\frac{-u_{,\theta}\; \sin\theta\; \rho^2+2u\cos\theta\; \rho^2-a^2 u\; \sin^2\theta  \cos\theta}{a \sin^3 \theta \rho^4}\;,$$
and
$$ ds^\flat(\alpha^\sharp, s)= \frac{1}{\sqrt{\rho^2}}\;\big[s^t \;\partial_\theta s_t + s^\varphi\; \partial_\theta s_\varphi\big]$$
$$=\frac{\cos\theta\;(r^2+a^2)}{\sin \theta\; (\rho^2)^{3/2}}\;.$$ Putting all the terms together in eq.(\ref{jfinal}), we find that, here
$$j=\frac{\sin\theta\;u_{,\theta}-u \cos\theta}{a \sin^3 \theta\;\rho^2} \;\partial_{r^\star}\;,$$
which can indeed be verified by a direct computation using eq.(\ref{Flsv1}), and (\ref{inhomMaxform}).
\subsection*{The Largest Class of Solutions Generated by $\partial_{r^\star}$}
As in the FLRW case, here too we see that $l=\partial_{r^\star}$ admits a uniform equipartition of null mean curvature:
$$g(\nabla_{\alpha^\sharp} \;\partial_{r^\star},s ) = \frac{a \cos\theta}{\rho^2} = -g(\nabla_s \;\partial_{r^\star}, \alpha^\sharp)\;.$$
Also, when
$$\hat s = A\; s+ B\;\alpha^\sharp\;,$$
we have that $\partial_{r^\star}$ and $\hat s$ form an involutive distribution
if and only if $A=A(t^\star, \theta, \varphi^\star)$ and $B=B(t^\star, \theta, \varphi^\star)$.
Therefore, theorem \ref{hatsoltheo} gives us that for any choice of such functions $A$ and $B$ such that $A^2+B^2 =1$, we have will have a new class of null force-free solutions, $\hat F_{l\hat s}$ of the form given by eq. (\ref{hatsol}), wherein
$$\ker \hat F_{l\hat s} =\; {\rm span}\;\;\{\partial_{r^\star}, \;A\; s+ B\;\alpha^\sharp\}\;.$$
Also,
$$d\hat s^\flat (l, \hat \alpha^\sharp)=\hat \alpha^t \;\partial_{r^\star} (\hat s_t) + \hat \alpha^\theta \;\partial_{r^\star} (\hat s_\theta) + \hat \alpha^\varphi \;\partial_{r^\star} (\hat s_\varphi)$$
$$=\frac{AB}{\sqrt{\rho^2}} \;\left(a^2 \sin^2 \theta\; \partial_{r^\star}\left(\frac{1}{\sqrt{\rho^2}}\right)-\partial_{r^\star}\left(\frac{(r^2 + a^2)}{\sqrt{\rho^2}}\right)\right)$$
$$+\frac{AB}{\sqrt{\rho^2}}\;\partial_{r^\star}\sqrt{\rho^2}=0\;.$$
Therefore, from theorem \ref{currentcomp}, regardless of choice of functions $A$ and $B$ the current density vector for $\hat F_{l\hat s}$  is along $\partial_{r^\star}$.
From eq.(\ref{hatsol}) and (\ref{lrot}) we get that
$$\hat F_{l\hat s} = (u \cdot \kappa)\; (-B \;s^\flat + A\; \alpha) \wedge l^\flat$$
$$= (u \cdot \kappa)\; \sqrt{\rho^2}\;( A\; d\theta-B \;\sin \theta\; d\varphi^\star) \wedge l^\flat\;.$$
Important point: $(u \cdot \kappa)\; \sqrt{\rho^2}$ must be $r^\star$ independent so that $d\hat F_{l\hat s}=0$. Subsuming this term into $A$ and $B$ we finally get that
\begin{equation}
 \hat F_{l\hat s}= \big( A \;d\theta-B \sin \theta\; d\varphi^\star\big) \wedge l^\flat\;.   \label{genloutkerrsol}
\end{equation}
This is exactly in the form given in \cite{GT14}, and so we do not analyze the solution any further. Suffice it to say that the solution given by eq.(\ref{genloutkerrsol}) is the largest class of solutions possible when $\partial_{r^\star}$ is in the kernel of a null field. We are able to make this assertion based on the structural aspects of the theory of null, force-free fields rather than by relying on the method of computational exhaustion.

\section{Conclusion}
In this work, we have stated the precise geometric conditions  under which a null force-free electromagnetic field may exist in curved spacetime. It is shown that the integral submanifolds of the kernel of the field tensor must contain a null geodesic. Further, these submanifolds alone must determine the expansion scalar of the null geodesic congruence. As a result, when null foliations admitting the equipartition of null curvature exists, then and then alone will spacetime admit a class of local null force-free solutions.

From a theoretical point of view, there is still an open question regarding how one may patch local solutions into meaningful global solutions. However, in astrophysical settings, this is not a large concern. Of chief importance in black hole astrophysics is the Kerr/Schwarzschild metric. In both these cases, we have a horizon penetrating coordinate system that is also valid everywhere in the external geometry. Therefore, once we establish a null field sheet foliation, we can easily write down the form of the field, as given by eq.(\ref{mixedF}), using a single horizon penetrating coordinate system. The only remaining variable, given by $u \cdot \kappa$, can be obtained by imposing the force-free Maxwell's equation. The latter part is a simple computation since we are guaranteed a solution. 

The formalism we have developed is not with the intention of making computations easier, it is rather to promote an understanding of the structural aspects of force-free, null fields. However, it might offer a new path in obtaining both analytical and numerical solutions. Keeping astrophysical relevance in mind, a possible new path in the case of the Kerr metric can be sketched. In this case, we already have an exact expression for any null geodesic tangent vector. The tangent vector $l$ is uniquely described by fixing the energy $\epsilon$, angular momentum $L$, and the Carter constant $K$ of the geodesic. In order to obtain the most general null geodesic congruence in Kerr geometry, we must elevate these integral constants to functions in spacetime such that
$$l(\epsilon)=l(L)=l(K)=0\;.$$
The above equation will ensure that each geodesic of the congruence will have a constant value for $\epsilon, L$, and $K$. For any such choice of a null congruence, all that remains, is to obtain $l^\perp$ and look for foliations admitting an equipartition of null mean curvature with respect to the null congruence $l$. If the path described above is not any easier than the existing techniques, it certainly does offer a new one. This methodology is generic and does not rely on the details of the Kerr metric as was made evident by constructing a specific example in an FLRW spacetime.
\bibliography{bibliography}

\end{document}